\documentclass{article}

\usepackage{graphicx,amsmath,amsthm,amssymb,bm,natbib}

\theoremstyle{definition}
\newtheorem{theorem}{Theorem}
\newtheorem{corollary}[theorem]{Corollary}
\newtheorem{lemma}{Lemma}
\newtheorem{remark}{Remark}

\theoremstyle{definition}

\newtheorem{example}{Example}

\newcommand{\AK}{}
\newcommand{\SEI}{}
\newcommand{\AKL}{}

\begin{document}

\title{Calculating the normalising constant of the Bingham distribution 
on the sphere using the holonomic gradient method}

\author{
Tomonari Sei \and Alfred Kume
}

\author{Tomonari Sei\footnote{Department of Mathematics, Keio University, Japan.}
\ and Alfred Kume\footnote{SMSAS, University of Kent, UK.} }

\date{}
\maketitle

\begin{abstract}
In this paper we implement the holonomic gradient method to exactly compute the normalising constant of Bingham distributions. This idea is originally applied for general Fisher-Bingham distributions in \cite{hgd}. In this paper we explicitly apply this algorithm to show the exact calculation of the normalising constant; derive explicitly the Pfaffian system for this parametric case; implement the general approach for the maximum likelihood solution search and finally adjust the method for degenerate cases, namely when the parameter values have multiplicities.

\end{abstract}

\section{Introduction}

Let $p\geq 2$ and $S^{p-1}=\{x\in \mathbb{R}^p\mid x^{\top}x=1\}$,
the unit sphere in the $p$-dimensional Euclidean space.
\SEI{Let $dx$ be the uniform measure on $S^{p-1}$ with $\int_{S^{p-1}}dx=2\pi^{p/2}/\Gamma(p/2)$.}
Then modulo an orthogonal transformation in $S^{p-1}$, the Bingham distribution has density function with respect to $dx$ on $S^{p-1}$ as
\begin{align}
 f(x|\theta)
 &= \frac{1}{C(\theta)}e^{\sum_{i=1}^p\theta_ix_i^2},
\end{align}
where $\theta=(\theta_1,\ldots,\theta_p)^{\top}$
is the parameter and $C(\theta)$ is the normalising constant
\begin{align}
 C(\theta)
 &= \int_{S^{p-1}} e^{\sum_{i=1}^p\theta_i x_i^2} dx.
 \label{eq:Bingham}
\end{align}

Simple arguments confirm that for any $c \in \mathbb{R}$
\[
 C(\theta) e^{c}=  C(\theta+c)
\]
where $\theta+c=(\theta_1+c,\ldots,\theta_p+c)^{\top}$. Hence without loss of generality we can assume that $\theta_i$ can be all positive.

\cite{KumeWood2005} show that $C(\theta)$ is actually closely related to a particular value of the density of a random variable defined as a linear combination of the $p$ independent $\chi_1^2$ random variables. In particular, using the Laplace-transform inversion arguments one can show the following one-dimensional representation of $C(\theta)$ is useful (\cite{KumeWood2005}).
\SEI{
\begin{align}
C(\theta)
 = \frac{C(0)}{2\pi i}\int_{t_0-i\infty}^{t_0+i\infty}
 \frac{1}{\prod_{k=1}^p\sqrt{-\theta_k-t}}e^{-t} dt,
 \label{eq:Bingham-1dim}
\end{align}
}where $t_0$ is any real number less than $-\theta_k$ for all $k$. \SEI{Recall that $C(0)=\int_{S^{p-1}}dx=2\pi^{p/2}/\Gamma(p/2)$.}
If $p=2$, the integral in (\ref{eq:Bingham-1dim}) should be interpreted appropriately because it is not integrable in the Lebesgue sense (see Appendix). \AK{Note however that for the $p=2$ case the normalizing constant is related to that of the von Mises-Fisher distribution involving the Bessel function of the first kind}.
\AK{While in general there is not a closed form for $C(\theta)$,  its calculation  is essential in likelihood estimation of the parameters of Bingham distributions.} The saddlepoint approximation is shown to work very well in a range of parameter values (see \cite{KumeWood2005}). In this paper however, we will exploit the connection between partial derivatives of $C(\theta)$ to implement the theory of differential equations. \AK{The basic idea here is that provided that we have a well defined curve in the parameter space $\theta$ whose value $C(\theta_0)$ at initial point $\theta_0$ is accurately known then the numerical methods of the differential theory will provide accurate solutions for the end point of the curve. In principle, if we can then provide the starting point accurately we will get the end point after numerical routines implementation. }
This approach has started to be implemented for similar distributions in 
\cite{hgd}. In this paper we will explicitly adopt the theory for the Bingham distribution by constructing the relevant Pfaffian equation and using it for deriving numerically the exact solution of $C(\theta)$.

The paper is organised as follows. \AK{In section 2 we provide a quick review of the Holonomic gradient methods. In \AKL{s}ection 3, we provide the necessary calculations for implementing this particular gradient method to Bingham distributions including the degenerate cases of multiplicities in the parameters. In section 4 we provide the numerical evidence of the method proposed. We then compare it with the saddle point approximation and other cases when we know the normalizing constant expression exactly. We conclude the paper with some discussion.}

\section{Review of the holonomic gradient methods}

In this section we review the framework of the holonomic gradient methods.
See \cite{hgd}, \cite{HNTT}, \cite{SO3}, \cite{Koyama2011}, \cite{yama4} and \cite{yama4b}
for details and further information.

We consider not only the Bingham distribution
but also a general parametric family
$f(x|\theta)$ on the sample space $\mathcal{X}$ with the parameter $\theta=(\theta_1,\ldots,\theta_d)^{\top}$.
The parameter space $\Theta$ is an open subset of the $d$-dimensional Euclidean space.

We assume that the density function $f(x|\theta)$
is an elementary function of $\theta$ and
\AKL{a} $r$-dimensional vector $G=G(\theta)$ satisfying the following PDE:
\begin{align}
 \partial_i G(\theta)
 = P_i(\theta)G(\theta),\quad i=1,\ldots,d,
 \label{eq:Pfaffian}
\end{align}
where 
$\partial_i$ denotes $\partial/\partial\theta_i$ and
$P_i(\theta)$ is a $r\times r$-matrix of rational functions of $\theta$.
The equation (\ref{eq:Pfaffian}) is called
{\em the Pfaffian equation} of $G$
and plays an essential role in this paper.
Typically the vector $G$ consists of the normalising constant of $f(x|\theta)$
and its derivatives.
We give an example.

\begin{example}
 Consider the von Mises--Fisher distribution
 \begin{align*}
  f(x|\theta,\mu)
  = \frac{1}{C(\theta)}e^{\theta\mu^{\top}x},
  \quad
  C(\theta)
  = \int_{S^{p-1}} e^{\theta\mu^{\top}x}dx,
 \end{align*}
 on the unit sphere $S^{p-1}$,
 where $\theta\geq 0$, $\mu\in S^{p-1}$,
 and $dx$ is the uniform \SEI{measure}.
 It is known that \SEI{$C(\theta)=(2\pi)^{p/2}I_{p/2-1}(\theta)/\theta^{p/2-1}$} 
 and $I_{\nu}$ denotes the modified Bessel function of the first kind and order $\nu$
 (see p.~168 of \cite{MardiaJupp2000}).
 The function $C$ satisfies the following ordinary differential equation:
 \begin{align*}
  C''(\theta) + \frac{p-1}{\theta}C'(\theta) - C(\theta) &= 0.
 \end{align*}
 Putting $\theta=(\theta_1)$ and $G=(C,C')^{\top}$,
 we have the Pfaffian equation:
 \begin{align*}
  &
  \partial_1
  \begin{pmatrix}G_1\\ G_2\end{pmatrix}
  = \begin{pmatrix}
     0& 1\\
     1& -\frac{p-1}{\theta}
    \end{pmatrix}
  \begin{pmatrix}G_1\\ G_2\end{pmatrix}.
 \end{align*}
 The density function is written as an elementary function
 $f(x|\theta,\mu) = e^{\theta\mu^{\top}x}/G_1(\theta)$
 of $\theta$, $\mu$ and $G$.
 Note that the modified Bessel function itself
 is not an elementary function.
 \qed
\end{example}

\subsection{The HG algorithm}

Assume that a numerical value of the vector $G(\theta^{(0)})$
at some point $\theta^{(0)}\in\Theta$ is given.
The {\em holonomic gradient (HG) algorithm} evaluates $G(\theta^{(1)})$ at
any other point $\theta^{(1)}$.
The algorithm is based on the following lemma.

\begin{lemma}
 Let $\bar\theta(\tau)$, $\tau\in[0,1]$,
 be a smooth curve in $\Theta$ such that
 $\bar\theta(0)=\theta^{(0)}$ and $\bar\theta(1)=\theta^{(1)}$.
 Put $\bar{G}(\tau)=G(\bar\theta(\tau))$.
 Then $\bar{G}(\tau)$ is the solution of
 the ordinary differential equation (ODE)
 \begin{align}
  \frac{d}{d\tau}\bar{G}(\tau)
  &= \sum_{i=1}^d \frac{d\bar\theta_i(\tau)}{d\tau}
  P_i(\bar\theta(\tau))\bar{G}(\tau),
  \label{eq:G-update}
 \end{align}
 with the initial condition $\bar{G}(0)=G(\theta^{(0)})$.
 In particular, $\bar{G}(1)=G(\theta^{(1)})$.
\end{lemma}

\begin{proof}
 By differentiation of composite functions, we have
 \begin{align*}
  \frac{dG(\bar{\theta}(\tau))}{d\tau}
  &= \sum_{i=1}^d \frac{d\bar{\theta}_i}{d\tau}\partial_iG(\bar{\theta}(\tau))
  \\
  &= \sum_{i=1}^d \frac{d\bar{\theta}_i}{d\tau}P_i(\bar{\theta}(\tau))G(\bar{\theta}(\tau)),
 \end{align*}
 where the last equality uses the Pfaffian equation (\ref{eq:Pfaffian}).
 This proves (\ref{eq:G-update}).
 The initial condition is obvious.
\end{proof}

The HG algorithm is described as follows.
A natural choice of $\bar{\theta}$
is the segment $\bar\theta(\tau)=(1-\tau)\theta^{(0)}+\tau\theta^{(1)}$
connecting $\theta^{(0)}$ and $\theta^{(1)}$.

\begin{description}
 \item[Input] $\theta^{(0)}$, $G(\theta^{(0)})$ and $\theta^{(1)}$.
 \item[Output] $G(\theta^{(1)})$.
 \item[Algorithm]
\end{description}
\begin{enumerate}
 \item Numerically solve the ODE (\ref{eq:G-update}) over $\tau\in[0,1]$.
 \item Return $\bar{G}(1)$.
\end{enumerate}
\AK{Note that the standard numerical routines for solving \eqref{eq:G-update} are highly accurate and available in most computer packages.}

\subsection{The discrete-time HGD algorithm}
\AK{In the following, we will implement the HG algorithm for maximum likelihood estimation of the parameters for the parametric family $f(x|\theta)$  (including that of Bingham).
Let $x(1),\ldots,x(N)$ be some observed data and we want to perform MLE based on the parametric family.}
The log-likelihood function $\ell(\theta)=\sum_{t=1}^N\log f(x(t)|\theta)$ is
written as $\ell(\theta)=L(\theta,G(\theta))$,
where $L:\Theta\times\mathbb{R}^r\to\mathbb{R}$ is an elementary function and $G$ \AK{satisfying} the Pfaffian equation (\ref{eq:Pfaffian}).
We consider the (naive) Newton-Raphson method:
\begin{align}
 \theta^{(k+1)}
 &= \theta^{(k)} - [\mathop{\rm Hess}\ell(\theta^{(k)})]^{-1}\mathop{\rm grad}\ell(\theta^{(k)}).
 \label{eq:NR}
\end{align}
where
\[
\mathop{\rm grad}\ell(\theta)=(\partial_i\ell(\theta))_{i=1}^d
\]
and
\[
\mathop{\rm Hess}\ell(\theta)=(\partial_i\partial_j\ell(\theta))_{i,j=1}^d.
\]
It is expected that the solution $\theta^{(k)}$ converges to
the MLE $\hat\theta$ as $k\to\infty$.

The {\em holonomic gradient descent (HGD) algorithm} numerically updates
$G(\theta^{(k)})$ by the HG algorithm.
The gradient vector and Hessian matrix of $\ell(\theta)$ at $\theta^{(k)}$
are computed only in terms of $\theta^{(k)}$ and $G(\theta^{(k)})$.
\AK{In fact $G$ is a vector of the normalizing constant and its derivatives which are closely related to derivatives in $\eqref{eq:NR}$.
More explicitly, the following lemma holds.}

\begin{lemma} \label{lemma:d_loglike}
 For any $i$ and $j$ in $\{1,\ldots,d\}$, we have
 \begin{align*}
  \partial_i \ell(\theta)
  &= \left.\frac{\partial L}{\partial\theta_i}
  + \sum_{a=1}^r(P_iG)_a\frac{\partial L}{\partial G_a}\right|_{G=G(\theta)},
  \\
  \partial_i\partial_j\ell(\theta)
  &= \frac{\partial^2L}{\partial\theta_i\partial\theta_j}
  + \sum_{a=1}^r(P_iG)_a\frac{\partial^2L}{\partial\theta_j\partial G_a}
  \\
  &\quad
  + \sum_{a=1}^r(P_jG)_a\frac{\partial^2L}{\partial\theta_i\partial G_a}
  \\
  &\quad 
  + \sum_{a=1}^r((\partial_jP_i)G+P_iP_jG)_a\frac{\partial L}{\partial G_a}
  \\
  &\quad
  + \left.\sum_{a,b=1}^r(P_iG)_a(P_jG)_b\frac{\partial^2L}{\partial G_a\partial G_b}\right|_{G=G(\theta)},
 \end{align*}
 where $(P_iG)_a$ is the $a$-th component of the $r$-vector $P_iG$ and so on.
\end{lemma}

\begin{proof}
 The formulas are obtained by
 differentiation of composite functions
 and the Pfaffian equation (\ref{eq:Pfaffian}).
\end{proof}

The holonomic gradient descent algorithm is described as follows.
We refer to this algorithm as {\em the discrete-time HGD algorithm}
in order to distinguish \AKL{it from} the continuous-time HGD algorithm
defined in the following section.

\begin{description}
 \item[Input] The initial point \SEI{$\theta^{(0)}\in\mathbb{R}^d$
	    and $G(\theta^{(0)})\in\mathbb{R}^r$.}
 \item[Output] The MLE $\hat\theta$ and $G(\hat\theta)$.
 \item[Algorithm] 
\end{description}
\begin{enumerate}
 \item Let $k=0$.
 \item Compute $\theta^{(k+1)}$ by the formula
       (\ref{eq:NR}) via Lemma~\ref{lemma:d_loglike}.
 \item Compute $G(\theta^{(k+1)})$ by the HG algorithm,
       i.e., numerically solve the ODE (\ref{eq:G-update})
       with $\bar\theta(\tau)=(1-\tau)\theta^{(k)}+\tau\theta^{(k+1)}$.
 \item If the norm $\|\mathop{\rm grad}\ell(\theta)\|$
       is sufficiently small,
       then return $\theta^{(k+1)}$ as output.
       Otherwise, put $k\leftarrow k+1$ and go to Step 2.
\end{enumerate}

\subsection{The continuous-time HGD algorithm}

We can also consider a continous-time version of the Newton-Raphson scheme.
Let $\ell(\theta)$ be the log-likelihood function as in the preceding section.
Consider the solution $\bar\theta(\tau)$ of the following differential equation:
\begin{align}
 \frac{d\bar\theta}{d\tau}
 &= -\frac{1}{1-\tau}[\mathop{\rm Hess}\ell(\bar\theta)]^{-1}\mathop{\rm grad}\ell(\bar\theta),
 \ \ 0\leq \tau<1,
 \label{eq:NR-cont}
\end{align}
where the gradient vector and Hessian matrix of $\ell(\theta)$
are calculated \AK{as in} Lemma~\ref{lemma:d_loglike}.
The solution $\bar\theta(\tau)$ of (\ref{eq:NR-cont}) converges to
the (local) maximum likelihood estimate $\hat\theta$ as $\tau\to 1$
because
\begin{align*}
 \frac{d}{d\tau}[\mathop{\rm grad}\ell(\bar\theta(\tau))]
 &= \mathop{\rm Hess}\ell(\bar\theta(\tau))\frac{d\bar\theta}{d\tau}
 = \frac{-1}{1-\tau}\mathop{\rm grad}\ell(\bar\theta(\tau))
\end{align*}
by (\ref{eq:NR-cont}).
This means
\begin{align}
 \mathop{\rm grad}\ell(\bar\theta(\tau))
 =(1-\tau)\mathop{\rm grad}\ell(\bar\theta(0)),
 \label{eq:NR-cont-grad}
\end{align}
which converges to $0$ as $\tau\to 1$.
The update rule for $\bar{G}(\tau):=G(\bar\theta(\tau))$
is the same as the equation (\ref{eq:G-update}).

The continuous-time version of the holonomic gradient descent is
described as follows.

\begin{description}
 \item[Input] The initial point \SEI{$\theta^{(0)}\in\mathbb{R}^d$, $G(\theta^{(0)})\in\mathbb{R}^r$}
	    and a sufficiently small number $\varepsilon>0$.
 \item[Output] The MLE $\hat\theta$ and $G(\hat\theta)$.
 \item[Algorithm]
\end{description}
\begin{itemize}
 \item[1] Solve the ordinary differential equation
	  (\ref{eq:G-update}) and (\ref{eq:NR-cont})
	  from $\tau=0$ to $\tau=1-\varepsilon$, simultaneously,
	  with the initial value $\bar\theta(0)=\theta^{(0)}$
	  and $\bar{G}(0)=G(\theta^{(0)})$.
 \item[2] Return $\bar\theta(1-\varepsilon)$ as $\hat{\theta}$.
\end{itemize}

\begin{remark}
 The continuous-time version is eaisier to implement
 than the discrete-time one.
 A drawback of the continuous version is that the point $\tau=1$
 of the equation (\ref{eq:NR-cont}) is singular.
 Therefore the number $\varepsilon$ in the HGD algorithm cannot be very small.
 Instead, the approximate value $\bar\theta(1-\varepsilon)$
 will be improved by one or several steps of the discrete-time scheme
 (\ref{eq:NR}).
\end{remark}

\section{Pfaffian equation for the Bingham distribution}

In this section, 
we consider the Bingham distribution (\ref{eq:Bingham}) again.
The Pfaffian equation for the Bingham distribution
is \AK{initially} derived for a generic case
and \AK{then we focus on the singular cases which occur when parameter values of $\theta$ have multiplicities.
For example, the complex Bingham distribution is a singular \AKL{case}.}

\subsection{Generic case}

We first derive partial differential equations satisfied
by the normalising constant of the Bingham distribution
on the set $\{\theta\mid \theta_i\neq \theta_j\ \mbox{for}\ i\neq j\}$.
Denote the partial derivative $\partial/\partial\theta_i$ by $\partial_i$.

\begin{lemma}
The normalising constant of the Bingham distribution satisfies
the following set of partial differential equations:
\begin{align}
 &\sum_{i=1}^p\partial_i C = C,
 \label{eq:PDE-1}
 \\
 &\partial_i\partial_jC = \frac{\partial_iC - \partial_jC}{2(\theta_i-\theta_j)},
 \quad 1\leq i<j\leq p.
 \label{eq:PDE-2}
\end{align}
\end{lemma}

Equation (\ref{eq:PDE-2}) is known as
the Euler-Darboux (or Euler-Poisson-Darboux) equation \SEI{(e.g.\ \cite{Takayama1992})}.

\begin{proof}
 Equation (\ref{eq:PDE-1}) immediately follows from the definition (\ref{eq:Bingham}) of $C(\theta)$.
 Indeed, since the support of the measure $dx$ is $S^{p-1}$, we have
 \begin{align*}
 \sum_{i=1}^p\partial_iC
 &= \int_{S^{p-1}} \left(\sum_{i=1}^p x_i^2\right)e^{\sum_{i=1}^p\theta_i x_i^2} dx
  \\
 & = \int_{S^{p-1}} e^{\sum_{i=1}^p\theta_i x_i^2} dx
 \\
 & = C.
 \end{align*}
 Equation (\ref{eq:PDE-2}) is proved as follows.
 Assume $p\geq 3$ for simplicity.
\AK{Based on} (\ref{eq:Bingham-1dim}) and Lebesgue's convergence theorem,
 the first derivative of $C$ is
 \SEI{
 \begin{align*}
  \partial_iC
  &= \frac{C(0)}{2\pi i}\int_{t_0-i\infty}^{t_0+i\infty}
  \frac{1}{2(-\theta_i-t)\prod_{k=1}^p\sqrt{-\theta_k-t}}e^{-t}dt.
 \end{align*}
}For $i\neq j$, the second derivative is
\SEI{
 \begin{align*}
  \partial_i\partial_jC
  &= \frac{C(0)}{2\pi i}\int_{t_0-i\infty}^{t_0+i\infty}
  \\
  &\frac{1}{4(-\theta_i-t)(-\theta_j-t)\prod_{k=1}^p\sqrt{-\theta_k-t}}e^{-t}dt.
 \end{align*}
}The partial fractional decomposition yields
 \begin{align*}
  &\frac{1}{4(-\theta_i-t)(-\theta_j-t)}
  \\
  &= \frac{1}{2(\theta_i-\theta_j)}
  \left(\frac{1}{2(-\theta_i-t)}-\frac{1}{2(-\theta_j-t)}\right).
 \end{align*}
 Then we obtain (\ref{eq:PDE-2}).
 For $p=2$, the same equation is derived.
\end{proof}

We give the Pfaffian equation of $C(\theta)$.
Let
\begin{align*}
 G=G(\theta)=(\partial_1C,\ldots,\partial_pC).
\end{align*}
Note that $C$ is written as
$C=\sum_{i=1}^pG_i$ due to (\ref{eq:PDE-1}).

\begin{theorem} \label{theorem:Pfaffian-Bingham}
 The vector $G$ satisfies the following PDEs.
 For $i$ and $j$ in $\{1,\ldots,p\}$,
 \begin{align}
  \partial_i G_j = \frac{G_i-G_j}{2(\theta_i-\theta_j)},
  \quad i\neq j,
  \label{eq:Pfaffian-1}
 \end{align}
 and
 \begin{align}
  \partial_i G_i
  = G_i - \sum_{k\neq i}^p\frac{G_i-G_k}{2(\theta_i-\theta_k)},
  \label{eq:Pfaffian-2}
 \end{align}
 where $\sum_{k\neq i}^p$ denotes summation over
 $k\in\{1,\ldots,p\}\setminus\{i\}$.
\end{theorem}

\begin{proof}
 The equation (\ref{eq:Pfaffian-1}) follows from the definition of $G$
 and the equation (\ref{eq:PDE-2}).
 We prove (\ref{eq:Pfaffian-2}).
 By (\ref{eq:PDE-1}), $\partial_iG_i$ is written as
 \begin{align*}
  \partial_i G_i
  &= \partial_i^2 C
  \\
  &= \partial_i
  \left(C-\sum_{k\neq i}\partial_kC\right)
  \\
  &= G_i - \sum_{k\neq i}\partial_iG_k.
 \end{align*}
 Since $\partial_iG_k$ is given by (\ref{eq:Pfaffian-1}),
 we have (\ref{eq:Pfaffian-2}).
\end{proof}

For the Bingham distribution, the parameter $\theta$
has redundancy because the following identity holds for any real number $c$:
\begin{align}
 f(x|\theta) = f(x|\theta-c1_p),
 \label{eq:redundancy}
\end{align}
where $1_p$ is the $p$-vector of ones.
There are two methods to remove the redundancy.
One of them is to restrict \AKL{one of} the parameter \AKL{say} $\theta_p=0$.
Then the derivative $\partial_p$ has to be removed
from the Pfaffian equation in a proper way.
The result is the following corollary.

\begin{corollary} \label{corollary:Pfaffian-Bingham-tilde}
 The vector
 \[
 \tilde{G}=(\tilde{G}_1,\ldots,\tilde{G}_p)=(C,\partial_1C,\ldots,\partial_{p-1}C)|_{\theta_p=0}
 \]
 satisfies the following equations.
 For $i$ and $j$ in $\{1,\ldots,p-1\}$,
 \begin{align}
  \partial_i \tilde{G}_1 = \tilde{G}_{i+1},
  \quad
  \partial_i \tilde{G}_{j+1} = \frac{\tilde{G}_{i+1}-\tilde{G}_{j+1}}{2(\theta_i-\theta_j)},
  \quad i\neq j,
  \label{eq:Pfaffian-3}
 \end{align}
 and
 \begin{align}
  \partial_i \tilde{G}_{i+1}
  &= \tilde{G}_{i+1} - \sum_{k\neq i}^{p-1}\frac{\tilde{G}_{i+1}-\tilde{G}_{k+1}}{2(\theta_i-\theta_k)}
  \nonumber\\
  &- \frac{\tilde{G}_{i+1}-(\tilde{G}_1-\sum_{\ell=1}^{p-1}\tilde{G}_{\ell+1})}{2\theta_i},
  \label{eq:Pfaffian-4}
 \end{align}
 where $\sum_{k\neq i}^{p-1}$ denotes
 the sum over $k\in\{1,\ldots,p-1\}\setminus\{i\}$.
\end{corollary}

\begin{proof}
 In (\ref{eq:Pfaffian-3}),
 the equation $\partial_i\tilde{G}_1=\tilde{G}_{i+1}$ for $1\leq i\leq p-1$
 is obvious from the definition of $\tilde{G}$.
 The equation for $\partial_i\tilde{G}_{j+1}$
 is the same as (\ref{eq:Pfaffian-1}).
 Finally, (\ref{eq:Pfaffian-4}) is obtained from (\ref{eq:Pfaffian-2})
 if one notes
 $G_p=\tilde{G}_1-\sum_{i=1}^{p-1}\tilde{G}_{i+1}$.
\end{proof}

Another method to remove the redundancy is to impose a penalty factor to the likelihood function
like
\begin{align*}
 \tilde{f}(x|\theta) = f(x|\theta)e^{-\theta_p^2/2}.
\end{align*}
Then the resultant \AK{MLE} should satisfy $\theta_p=0$.
Similarly, if we adopt
\begin{align*}
 \tilde{f}(x|\theta) = f(x|\theta)e^{-(\theta_1+\cdots+\theta_p)^2/2},
\end{align*}
then the resultant \AK{MLE} satisfies $\sum_{i=1}^p\theta_i=0$.
The penalty factor can be considered as the Bayesian prior density for $\theta$.

\subsection{Degenerate case}

We \AK{now} consider \AK{the} degenerate points of equations (\ref{eq:Pfaffian-1}) and (\ref{eq:Pfaffian-2}).
If \AK{multiplicities occur, i.e. some $\theta_i=\theta_j$ for some  $i\neq j$}, coincide,
then the Pfaffian equation can not be defined because the demonimator in (\ref{eq:Pfaffian-1}) becomes zero.
Connections between \AK{such} degenerate cases and those \AK{in the generic case} are essentially higher order derivatives \AK{(see e.g. \citep{KumeWood2007})}. To deal with such degenerate cases, we give the following theorem.

\begin{theorem} \label{theorem:Pfaffian-degenerate}
 Let $\theta=(\phi_1,\ldots,\phi_1,\ldots,\phi_q,\ldots,\phi_q)$,
 where $\{\phi_j\}$ takes distinct values and $\phi_j$ appears $d_j$ times for each $j=1,\ldots,q$.
 Denote the derivative $\partial/\partial\phi_j$ by $\partial_{\phi_j}$.
 Then we have the following equation of
 \[
 G:=(\partial_{\phi_1}C(\theta),\ldots,\partial_{\phi_q}C(\theta)).
 \]
 For $i$ and $j$ in $\{1,\ldots,q\}$,
 \begin{align}
  \partial_{\phi_i}G_j &= \frac{d_jG_i-d_iG_j}{2(\phi_i-\phi_j)},
  \quad i\neq j,
  \label{eq:Pfaffian-degenerate-1}
  \\
  \partial_{\phi_i}G_i &= G_i - \sum_{k\neq i}^q \frac{d_kG_i-d_iG_k}{2(\phi_i-\phi_k)}.
  \label{eq:Pfaffian-degenerate-2}
 \end{align}
 In particular, if $q=p$ and $d_1=\cdots=d_p=1$,
 then the equations coincide with (\ref{eq:Pfaffian-1}) and (\ref{eq:Pfaffian-2}).
 Furthermore, $\tilde{G}=(\tilde{G}_1,\ldots,\tilde{G}_p)
 =(C,\partial_{\phi_1},\ldots,\partial_{\phi_{p-1}}C)|_{\phi_p=0}$
 satisfies the following equations.
 For $i$ and $j$ in $\{1,\ldots,p-1\}$,
 \begin{align}
 &\partial_{\phi_i}\tilde{G}_1
 = \tilde{G}_{i+1},
 \quad \partial_{\phi_i}\tilde{G}_{j+1}
 = \frac{d_i\tilde{G}_j - d_j\tilde{G}_i}{2(\phi_i-\phi_j)},
 \quad i\neq j,
 \label{eq:Pfaffian-degenerate-3}
 \\
 &\partial_{\phi_i}\tilde{G}_{i+1}
 = \tilde{G}_{i+1} - 
 \sum_{k\neq i}^{p-1}
 \frac{d_i\tilde{G}_{k+1} - d_k\tilde{G}_{i+1}}{2(\phi_i-\phi_k)}
 \nonumber\\
 &\quad
 - \frac{d_k\tilde{G}_{i+1} - d_i(\tilde{G}_1 - \sum_{\ell=1}^{p-1}\tilde{G}_{\ell+1})}{2\phi_i}.
  \label{eq:Pfaffian-degenerate-4}
 \end{align}
\end{theorem}

\begin{proof}
 These equations are derived in the same \AK{way} as Theorem~\ref{theorem:Pfaffian-Bingham}
 using the 1-dimensional representation (\ref{eq:Bingham-1dim}) of $C(\theta)$.
\end{proof}

As a corollary of the theorem,
we obtain the Pfaffian equation for {\em the complex Bingham distribution}.
The complex Bingham distribution is
a distribution on the complex sphere
$S_c^{2q-1}=\{x\in\mathbb{C}^q\mid x^{\dagger}x=1\}$,
where $x^{\dagger}$ denotes the complex conjugate transpose of the complex vector $x$.
The complex Bingham density function with respect to the uniform distribution is
\begin{align*}
 f_c(x|\phi)
 &= C_c(\phi)^{-1}e^{\sum_{i=1}^q \phi_i |x_i|^2},
\end{align*}
where $C_c(\phi)$ is the normalising constant $C_c(\phi)$.
It is known that
\SEI{
\begin{align*}
C_c(\phi) &= 2\pi^q\sum_{j=1}^q a_je^{\phi_j}, 
 \quad a_j^{-1} = \prod_{i\neq j}(-\phi_i+\phi_j)
\end{align*}
}and \AK{this is a special case of Bingham distributions where the entries in the $\theta$  parameter are in pairs i.e. $d_1=\cdots=d_q=2$ (see e.g. \AKL{\cite{Ken:94}} \ p.472 of \cite{KumeWood2005}).
We will compare the exact expression above with that of the HG algorithm in the next Section.}


From Theorem~\ref{theorem:Pfaffian-degenerate},
the PDE \AK{varies} depending on how the parameter \AK{vector $\theta$} is degenerated, namely the multiplicities in $\theta$.
However we \AK{show below} that, at least for the maximum likelihood estimation,
\AK{the multiplicities in $\theta$ do not change since they are driven by the sufficient statistics}.
Let $s_i=N^{-1}\sum_{t=1}^Nx_i(t)^2$, $i=1,\ldots,p$, be the sufficient 
statistics with respect to the Bingham density.
Then the log-likelkihood function is given by
\begin{align*}
 \ell(\theta) &= N\left(
 \sum_{i=1}^p\theta_is_i - \log C(\theta)
 \right).
\end{align*}
\SEI{From the theory of exponential families,
there exists the unique MLE up to the redundancy (\ref{eq:redundancy})
if and only if the sufficient statistics satisfy $s_i>0$ for $i=1,\ldots,p$
(see e.g.\ \cite{barndorff-nielsen}, Corollary 9.6).
Note that if $s_i=0$ for some $i$, then all the data $\{x_k(t)\}$
 lies on a common hyperplane $x_i=0$
 and the $i$-th coordinate can be removed from the analysis.
 We assume the condition $s_i>0$ for all $i$ hereafter.
}

We first give a lemma.

\begin{lemma} \label{lemma:order-preserved}
 Assume
 $s = (\sigma_1,\ldots,\sigma_1,\ldots,\sigma_q,\ldots,\sigma_q)$,
 where $\sigma_i$ appears $d_i$ times for each $i=1,\ldots,q$
 and satisfies $\sigma_1<\cdots<\sigma_q$.
 Then the maximum likelihood estimate 
 forms into $\hat\theta=(\hat\phi_1,\ldots,\hat\phi_1,\ldots,\hat\phi_q,\ldots,\hat\phi_q)$,
 where $\hat\phi_1<\cdots<\hat\phi_q$.
\end{lemma}

\begin{proof}  
\SEI{
 Let $\hat\theta=(\hat\theta_i)_{i=1}^p$ be the unique MLE.
 We first prove that $\hat\theta_i\leq\hat\theta_j$ for any $i<j$.
 Note that $s_i\leq s_j$ for any $i<j$ by the assumption.
 Assume $\hat\theta_i>\hat\theta_j$ for some $i<j$.
 Then define a permuted vector $\hat\theta^{\pi}$ by $\hat\theta^{\pi}_i=\hat\theta_j$,
 $\hat\theta^{\pi}_j=\hat\theta_i$ and $\hat\theta^{\pi}_k=\hat\theta_k$ for other $k$'s.
 Since $C(\hat\theta^{\pi})=C(\hat\theta)$ by symmetry, we have
 \begin{align*}
  \ell(\hat\theta^{\pi}) - \ell(\hat\theta)
  &= \sum_{k=1}^p \hat\theta_k^{\pi}s_k - \sum_{k=1}^p \hat\theta_ks_k
  \\
  &= (\hat\theta_i-\hat\theta_j)(s_j-s_i)
  \geq 0.
 \end{align*}
 This contradicts to uniqueness of the MLE.
 Hence $\hat\theta_i\leq \hat\theta_j$ for any $i<j$.
 The same argument shows that
 $\hat\theta_i=\hat\theta_j$ for any $i<j$ such that $s_i=s_j$.
 Finally, we prove $\hat\theta_i<\hat\theta_j$ for any $i<j$ such that $s_i<s_j$.
 Assume $\hat\theta_i=\hat\theta_j$ for such $i$ and $j$.
 Then, by using the likelihood equation, we have
 \[
 s_i
 = \frac{\partial}{\partial\theta_i} \log C(\hat\theta)
 = \frac{\partial}{\partial\theta_j} \log C(\hat\theta)
 = s_j,
 \]
 which contradicts to the assumption $s_i<s_j$.
 This completes the proof.
 }
\end{proof}

By the above lemma and symmetry with respect to indices,
the order of $\hat\theta_1,\ldots,\hat\theta_p$
coincides with that of $s_1,\ldots,s_p$ including the number of \AK{multiplicities.}
Furthermore, the following theorem states that
the order is preserved during the HGD algorithm.

\begin{theorem}
 Assume that the order of $\theta^{(0)}$ coincides with that of the sufficient statistics $s$
 (including \AK{multiplicities}).
 Then the order of $\bar\theta(\tau)$ is preserved over $\tau\in[0,1)$,
 where $\bar\theta(\tau)$ is the update rule (\ref{eq:NR-cont}) of the continuous-time HGD algorithm.
\end{theorem}

\begin{proof}
 Let $\eta(\theta)=\mathop{\rm grad}\log C(\theta)$ be the expectation parameter.
 \SEI{Then the update rule (\ref{eq:NR-cont-grad}), equivalent to (\ref{eq:NR-cont}), is written as
  \[
  s - \eta(\bar{\theta}(\tau)) = (1-\tau) \{s - \eta(\theta^{(0)})\},
  \]
  or
  \begin{align}
  \eta(\bar{\theta}(\tau)) = (1-\tau) \eta(\theta^{(0)}) + \tau s.
  \label{eq:eta-tau}
  \end{align}
  This is the line segment connecting $\eta(\theta^{(0)})$ and $s$.
  From the argument of exponential families,
  $\theta^{(0)}$ is considered as the MLE when the sufficient statistic is $\eta(\theta^{(0)})$.
  Therefore, by Lemma~\ref{lemma:order-preserved}, the order of $\theta^{(0)}$ coincides with $\eta(\theta^{(0)})$.
  By (\ref{eq:eta-tau}) and the assumption on the order of $\theta^{(0)}$,
  we deduce that the order of $\eta(\bar{\theta}(\tau))$ is preserved over $\tau\in[0,1)$.
  Finally, since $\bar{\theta}(\tau)$ is considered as the MLE when the sufficient statistic is $\eta(\bar{\theta}(\tau))$,
  the order of $\bar{\theta}(\tau)$ is the same as $\eta(\bar{\theta}(\tau))$ by Lemma~\ref{lemma:order-preserved} again.
  This proves the theorem.
 }
\end{proof}

\section{Implementation issues and numerical evidence}
\AK{In this section we focus on the performance of our method. Note that there are two numerical procedures which need to work here. One is the accurate initial condition of the PDE equation and secondly the accuracy of the solution obtained via the PDE machinery. The later is a standard issue in implementation of the relevant packages where controllable accuracy is possible. We focus on the first procedure here, that of accurate calculation of the initial values.}

\subsection{Initial values}

For the HG and HGD algorithms,
we need to compute $C(\theta^{(0)})$ and its derivatives
at an appropriate point $\theta^{(0)}$.
If $\theta^{(0)}$ is sufficiently small,
then \SEI{they are calculated by the following power series expansion with an appropriate truncation}.

\SEI{Let $\theta(\phi,d)=(\phi_1,\ldots,\phi_1,\ldots,\phi_q,\ldots,\phi_q)$
be a parameter vector,
where the multiplicity of $\phi_i$ is $d_i$ for each $i=1,\ldots,q$
and therefore $\sum_{i=1}^q d_i=p$.
Denote $\phi=(\phi_i)_{i=1}^q$ and $d=(d_i)_{i=1}^q$.
Then, by using the argument of \cite{KumeWood2007}, we have
\begin{align}
\label{eq:power-series}
 &C(\theta(\phi,d))
 = \frac{2\pi^{p/2}}{\prod_{i=1}^q\Gamma(\frac{d_i}{2})}
 \nonumber\\
 &\quad\times\sum_{k_1=0}^{\infty}\cdots\sum_{k_q=0}^{\infty}
 \frac{\phi_1^{k_1}\cdots\phi_q^{k_q}}{k_1!\cdots k_q!}
 \frac{\prod_{i=1}^q \Gamma(k_i+\frac{d_i}{2})}{\Gamma(\sum_{i=1}^q (k_i+\frac{d_i}{2}))}.
\end{align}
The derivatives are calculated by
\begin{align*}
 &\frac{\partial^{|j|}}{\partial\phi_1^{j_1}\cdots\partial\phi_q^{j_q}}C(\theta(\phi,d))
 \\
 &= \frac{\prod_{i=1}^q\Gamma(d_i+2j_i)}{\pi^{|j|}\prod_{i=1}^q\Gamma(d_i)}C(\theta(\phi,d+2j)).
\end{align*}
for any index $j=(j_1,\ldots,j_q)$, where $|j|=\sum_{i=1}^qj_i$
(\cite{KumeWood2007}).
Let
\begin{align*}
&C_N(\theta(\phi,d))
= \frac{2\pi^{p/2}}{\prod_{i=1}^q\Gamma(\frac{d_i}{2})}
\\
&\quad
 \times\sum_{k_1+\cdots+k_q< N}\frac{\phi_1^{k_1}\cdots\phi_q^{k_q}}{k_1!\cdots k_q!}
 \frac{\prod_{i=1}^q\Gamma(k_i+\frac{d_i}{2})}{\Gamma(\sum_{i=1}^q (k_i+\frac{d_i}{2}))}.
\end{align*}
}According to \cite{yama4},
we have
\SEI{
\begin{align}
\label{eq:accuracy-1}
&|C(\theta(\phi,d)) - C_N(\theta(\phi,d))|
\nonumber\\
&\quad \leq \frac{\|\phi\|_1^N}{N!} \frac{N+1}{N+1-\|\phi\|_1}C(0).
\end{align}
where $\|\phi\|_1=\sum_{i=1}^q|\phi_i|$.
}See Appendix \ref{section:truncation} for details.

Table~\ref{table:universal-value} shows the computed values of $\tilde{G}(\theta)$
(see Corollary~\ref{corollary:Pfaffian-Bingham-tilde} for the definition of $\tilde{G}(\theta)$)
at $\theta=((p-i)/2p)_{i=1}^p$ for each dimension $p$.
The truncation number $N$ is selected
to make the right hand sides of (\ref{eq:accuracy-1})
less than \AK{some required accuracy $\varepsilon$. In our numerical examples we have chosen $\varepsilon$ to be} $10^{-6}$.

\begin{table}[htbp]
\caption{A set of values $\tilde{G}(\theta)$ \SEI{normalized by $C(0)=2\pi^{p/2}/\Gamma(p/2)$} is shown,
where $\theta$ is $((p-i)/(2p))_{i=1}^p$.
The last column shows computational time [sec]
of the power series.}
\label{table:universal-value}
\begin{center}
 \begin{tabular}{|c|l|r|}
 \hline
  $p$& \multicolumn{1}{|c|}{\SEI{$\tilde{G}(\theta)/C(0)$}}& \multicolumn{1}{|c|}{time}\\
  \hline
  2& $(1.137579,\ 0.604270)$& 0.0\\
  3& $(1.185742,\ 0.421987,\ 0.394412)$& 0.0\\
  4& $(1.210162,\ 0.321833,\ 0.308437,\ 0.295857)$& 0.0\\
  5& $(1.224897,\ 0.259286,\ 0.251813,\ 0.244669,$& 0.2\\
  &$\ 0.237834)$&\\
  6& $(1.234745,\ 0.216746,\ 0.212168,\ 0.207741,$& 1.0\\
  &$\ 0.203460,\ 0.199319)$& \\
  7& $(1.241789,\ 0.186029,\ 0.183026,\ 0.180101,$& 5.3\\
  &$\ 0.177252,\ 0.174476,\ 0.171771)$& \\
  8& $(1.247075,\ 0.162847,\ 0.160774,\ 0.158744,$& 26.4\\
  &$\ 0.156756,\ 0.154810,\ 0.152903, 0.151036)$& \\
  9& $(1.251187,\ 0.144750,\ 0.143260,\ 0.141795,$&  127.5\\
  &$\ 0.140356,\ 0.138941,\ 0.137550,\ 0.136182,$&\\
  &$\ 0.134837)$&\\
10& $(1.254477,\ 0.130242,\ 0.129136,\ 0.128045,$& 589.5\\
  &$\ 0.126970,\ 0.125910,\ 0.124866,\ 0.123836,$& \\
  &$\ 0.122821,\ 0.121820)$& \\
   \hline
 \end{tabular}
\end{center}
\end{table}

The power series expansion (\ref{eq:power-series})
is \AK{clearly} valid for every $\theta$ \AK{while the number of terms needed for the series approximation will be heavily dependent on the norm of the parameter $\theta$.}
\AK{Hence}, if the norm of $\theta$ is large,
then the truncation number $N$ that assures \AK{the required}  accuracy \AK{is likely to  become intolerably} large.
\AK{Both} HG and HGD methods \AK{are very} useful for this situations.
Specifically, \AK{in order} to compute the values of $C(\theta)$ for large $\theta$,
we first compute $\tilde{G}(\theta^{(0)})$ for sufficiently small $\theta^{(0)}$
and then apply the HG method with $\theta^{(1)}=\theta$.
To avoid singular points of the Pfaffian system, the order of the components of $\theta^{(0)}$
(including \AK{multiplicities}) should coincide with \AK{those} of $\theta$.

The initial point $\theta^{(0)}$ of the HG algorithm can be arbitrarily selected
independent of $\theta^{(1)}$
\AK{as long as} it avoids the singular points.
In other words, \AK{we can run the HG algorithm for any possible choice of the generic parameter $\theta^{(1)}$,
while we keep the initial value fixed at a specific point $\theta^{(0)}$.
Hence for a given $p$, it suffices or the HG algorithm to provide in tabular form the relevant values $G(\theta_0)$ for some fixed $\theta^{(0)}$.  In Table~\ref{table:universal-value} we show
all necessary values for $2\leq p\leq 10$ for the generic case. Note that $C(\theta)=C(\pi\theta)$ for any permutation $\pi$.}

\subsection{The HG method} \label{subsection:HG}


Table~\ref{table:HG-timing} compares
computational time of
the power series expansion and the HG method.
The parameter values examined are
$\theta=(a(p-i)^b)_{i=1}^p$ for several $p$, $a$ and $b$.
For the HG method,
the initial point is selected to $\theta^{(0)}=\theta/2\|\theta\|_1$,
which means $\|\theta^{(0)}\|_1=1/2$,
and the initial value $\tilde{G}(\theta^{(0)})$ is calculated from the power series expansion.
For each case in Table~\ref{table:HG-timing},
numerical values of $C(\theta)$ computed by the two methods
coincide up to $10^{-6}$
whenever the power series method returns the value
in a practical time.

If $\theta$ is so large that the power series method fails,
then we examine a logarithmic version of the Pfaffian system
to confirm numerical accuracy indirectly.
Let $G^L = (\partial_1\log C,\ldots,\partial_p\log C)=G/C$.
Then $G^L$ satisfies the following non-linear system:
\begin{align}
\label{eq:logarithmic}
 \partial_i G^L 
 = (P_i G^L)_j - G_i^L G_j^L
\end{align}
for $i,j=1,\ldots,p$, where $P_i$ is the Pfaffian matrix
defined by (\ref{eq:Pfaffian-1}) and (\ref{eq:Pfaffian-2}).
This system is numerically more stable than that of $G$
because $G$ becomes quite large if $\theta$ is large.
Figure~\ref{fig:hg} is the trajectory of
$G(\bar{\theta}(\tau))$ and $G^L(\bar{\theta}(\tau))$ against $\tau\in[0,1]$, respectively,
when $\theta^{(0)} = (0.4,0.1,0)$ and $\theta^{(1)}=(20,5,0)$.

\begin{table}[htbp]
\caption{Computational time [sec] of the PS (power series) and HG algorithms.
The parameter values examined are $\theta=(a(p-i)^b)_{i=1}^p$.
The symbol NA means that the PS method did not return
an output in a practical time.
For such cases, the logarithmic version (\ref{eq:logarithmic}) is used
to confirm numerical accuracy.
}
\label{table:HG-timing}
\begin{center}
 \begin{tabular}{|ccc|r|r|r|r|}
  \hline
  $p$& $a$& $b$& \multicolumn{1}{|c|}{$\|\theta\|_1$}&\multicolumn{1}{|c|}{\SEI{$C(\theta)/C(0)$}}& \multicolumn{1}{|c|}{PS}& \multicolumn{1}{|c|}{HG}\\\hline
  $5$& $1/20$& $1$& $1/2$& 1.105961& 0.1& 0.3\\
  $5$& $1/10$& $1$& 1& 1.224897& 0.2& 0.3\\
  $5$& $1$& $1$& 10&9.769432& 17.1& 0.3\\
  $5$& $10$& $1$& 100&$3.824\times 10^{14}$& NA& 0.3\\
  $5$& $1/60$& $2$& $1/2$& 1.106713& 0.1& 0.3\\
  $5$& $1$& $2$& 30& $5.253880\times 10^4$& 48.6& 0.3\\
  \hline
$10$& $1/90$& $1$& $1/2$& 1.051360& 14.0& 14.8\\
$10$& $1/45$& $1$& 1& 1.105546& 49.7& 14.7\\
$10$& $2/45$& $1$& $2$& 1.223062& 386.2& 14.6\\
$10$& $1$& $1$& $45$& $1.757059\times 10^2$& NA& 14.6\\
$10$& $1/570$& $2$& $1/2$& $1.051466$& 13.9& 14.1\\
$10$& $1$& $2$& $285$& $3.802\times 10^{28}$& NA& 15.2\\
  \hline
 \end{tabular}
\end{center}
\end{table}

\begin{figure}[htbp]
 \begin{center}
 \begin{tabular}{c}
  \includegraphics[width=0.4\textwidth]{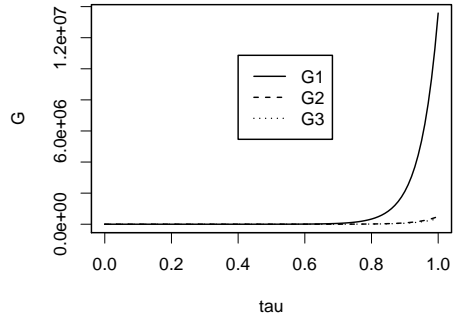}\\
  (a) Trajectory of $G(\bar{\theta}(\tau))$.\\
  \includegraphics[width=0.4\textwidth]{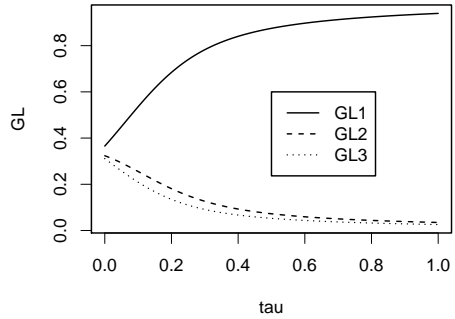}\\
  (b) Trajectory of $G^L(\bar{\theta}(\tau))$.
  \end{tabular}
 \end{center}
 \caption{Trajectory of $G(\bar{\theta}(\tau))$ and $G^L(\bar{\theta}(\tau))$ against $\tau\in[0,1]$,
 where $\bar{\theta}(\tau)$ is the line segment connecting
 $\theta^{(0)}=(0.4,0.1,0)$ and $\theta^{(1)}=(20,5,0)$.}
 \label{fig:hg}
\end{figure}

\subsection{Comparison with the saddle point approximation}

\AK{
One method which is used for likelihood inference on Bingham distributions is based on the saddle point approximations.} The first order saddle point approximation of (\ref{eq:Bingham-1dim}) is
\begin{align}
 C(\theta)
 &\approx
 \frac{1}{\sqrt{2\pi}} \frac{1}{\prod_{k=1}^p\sqrt{-\theta_k-t_*}}
 \nonumber\\
 &\left\{\sum_{\ell=1}^p\frac{1}{(-\theta_{\ell}-t_*)^2}\right\}^{-1/2}e^{-t_*},
 \label{eq:saddle}
\end{align}
where $t_*=t_*(\theta)$ is the unique solution of
\begin{align}
 \sum_{k=1}^p\frac{1}{-\theta_k-t_*} = 1
 \quad \mbox{and} \quad t_*<\min_k(-\theta_k).
\end{align}
\AK{Second order approximations are improvements of the one above (see \citep{KumeWood2005} for more details). It is shown however, that these improved versions are adequate for many practical applications since it takes a very large sample size so that the MLE estimates differ significantly from those of saddle point approximations. Another nice feature of the saddle point approximation is that the whole method is fast and it involves only a single one dimensional optimization procedure.  
In Tables ~\ref{tab:sadd:compl:comp:1} and ~\ref{tab:sadd:compl:comp:2} we compare the \AKL{second order }saddle point approximation of the normalizing constant with HG algorithm. We also compare both of these methods in the cases of Complex Bingham distributions, whose normalizing constant is known in closed form. As can be seen the HG algorithm performs well and is exact in the complex Bingham case.}
\begin{table}[ht]
\caption{Columns 2 and 3 in the table compare the saddle point approximations (spa) of \SEI{$\theta=(0,-1,-2,-\kappa)$} with that of hg algorithm; columns 4 and 5 compare the same quantities for \SEI{$\theta=(0,-1,-2,-\kappa,-\kappa)$} and the last three columns compare respectively the saddle point, exact and hg of the complex Bingham with parameters \SEI{$\phi=(0,-1,-2,-\kappa)$}}
\label{tab:sadd:compl:comp:1}
\begin{center}
\begin{tabular}{|r|rr|rr|rrr|}
  \hline
 $\kappa$ & spa & hg & spa & hg & spa & ex & hg \\ 
  \hline
5 & 4.237006 & 4.238950 & 3.376766 & 3.372017 & 5.942975 & 5.936835 & 5.936835 \\ 
  10 & 2.982628 & 2.985576 & 1.689684 & 1.689355 & 3.429004 & 3.425468 & 3.425468 \\ 
  30 & 1.708766 & 1.711919 & 0.555494 & 0.556123 & 1.248280 & 1.246421 & 1.246421 \\ 
  50 & 1.321178 & 1.323994 & 0.332102 & 0.332661 & 0.761347 & 0.760180 & 0.760180 \\ 
  100 & 0.932895 & 0.935094 & 0.165587 & 0.165940 & 0.385272 & 0.384675 & 0.384675 \\ 
  200 & 0.659185 & 0.660814 & 0.082676 & 0.082871 & 0.193779 & 0.193477 & 0.193477 \\ 
   \hline
\end{tabular}
\end{center}
\end{table}
\begin{table}[ht]
\caption{Columns 2 and 3 in the table compare the saddle point approximations (spa) of \SEI{$\theta=(0,-1,-22,-\kappa)$}; columns 4 and 5 compare the same quantities for \SEI{$\theta=(0,-1,-22,-\kappa,-\kappa)$} and the last three columns compare respectively the saddle point, exact and hg of the complex Bingham with parameters \SEI{$\phi=(0,-1,-22,-\kappa)$}}
\label{tab:sadd:compl:comp:2}
\begin{center}
\begin{tabular}{|r|rr|rr|rrr|}
  \hline
$\kappa$& spa & hg & spa & hg & spa & ex & hg \\
  \hline
5 & 1.258672 & 1.273161 & 1.032128 & 1.044072 & 0.921027 & 0.921726 & 0.921726 \\ 
  10 & 0.874523 & 0.883394 & 0.500707 & 0.505223 & 0.506236 & 0.506341 & 0.506341 \\ 
  30 & 0.497757 & 0.503213 & 0.162251 & 0.163901 & 0.177602 & 0.177495 & 0.177495 \\ 
  50 & 0.384440 & 0.388775 & 0.096784 & 0.097828 & 0.107526 & 0.107458 & 0.107458 \\ 
  100 & 0.271249 & 0.274375 & 0.048182 & 0.048725 & 0.054115 & 0.054081 & 0.054081 \\ 
  200 & 0.191595 & 0.193826 & 0.024039 & 0.024316 & 0.027144 & 0.027127 & 0.027127 \\ 
   \hline
\end{tabular}
\end{center}
\end{table}

\subsection{The HGD method}


Table~\ref{table:HGD-timing} compares
computational time of the discrete- and continuous-time HGD methods.
The data is $s=(2i / p(p+1))_{i=1}^p$ for each $p$.
The initial point $\theta^{(0)}=(\theta_1^{(0)},\ldots,\theta_p^{(0)})$ is selected such that
the order of $(s_1,\ldots,s_p)$ coincides with the order of $\theta^{(0)}$.
This avoids singularity.
The numerical error of the MLE $\hat{\theta}$ is evaluated by
$\max_{1\leq i\leq p} |\partial_i\log C(\hat\theta) - s_i|$,
which must be zero if $\hat\theta$ is correct.

\begin{table}[htbp]
 \caption{
 Computational time [sec] and numerical error
 of the HGD algorithm are shown.
 The data is $s=(2i/p(p+1))_{i=1}^p$ for each $p$.
 The error is evaluated by $\max_i |\partial_i\log C(\hat\theta) - s_i|$.
 Here $\partial_i\log C(\hat\theta)$ is obtained
by the power series expansion with accuracy $10^{-8}$
for $p\leq 5$.
Just for information, the error evaluated by the HG method
is displayed for $p\geq 6$.
}
 \label{table:HGD-timing}
 \begin{center}
  \begin{tabular}{ccccc}
  $p$& error& error& time& time\\
  & (discrete)& (continuous)& (discrete)& (continuous)\\
   \hline
 2& 2.41e-07& 1.04e-08& 0.01& 0.16\\
 3& 6.53e-07& 1.81e-08& 0.02& 0.21\\
 4& 5.40e-07& 1.41e-08& 0.05& 0.37\\
 5& 1.45e-06& 1.78e-08& 0.14& 0.60\\
 6& (1.69e-06)& (1.09e-08)& 0.44& 0.99\\
 7& (2.76e-06)& (1.17e-08)& 1.15& 1.76\\
 8& (6.14e-06)& (1.29e-08)& 2.82& 3.81\\
 9& (1.65e-05)& (2.29e-08)& 6.60& 7.64\\
10& (1.69e-05)& (2.06e-08)& 14.1& 15.7\\
   \hline
  \end{tabular}
 \end{center}
\end{table}

Figure~\ref{fig:hgd} shows
an example of the trajectories of $\theta^{(k)}$ (for discrete algorithm)
and $\bar{\theta}(\tau)$ (for continuous algorithm).
The data is
\begin{align*}
 s & = \left(\frac{1}{15}, \frac{2}{15}, \frac{3}{15}, \frac{4}{15}, \frac{5}{15}\right)
\end{align*}
and the MLE computed by the continuous-time HGD algorithm is
\begin{align*}
 \hat\theta &= (-7.188333,\ -3.120184,\ -1.543555,\ -0.628081,\ 0).
\end{align*}
The initial parameter is
\begin{align*}
 \theta^{(0)} &= \left( -\frac{4}{20},-\frac{3}{20}, -\frac{2}{20}, -\frac{1}{20}, 0 \right).
\end{align*}

\begin{figure}[htbp]
 \begin{center}
 \begin{tabular}{c}
  \includegraphics[width=0.4\textwidth]{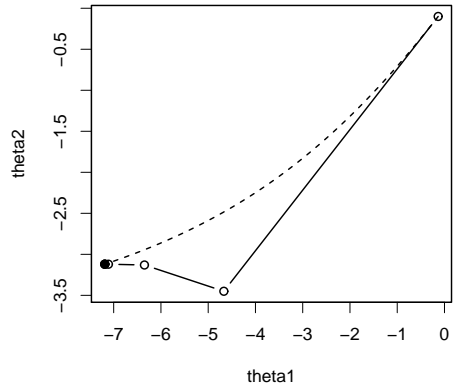}\\
  (a) $(\theta_1,\theta_2)$-plane.\\
  \includegraphics[width=0.4\textwidth]{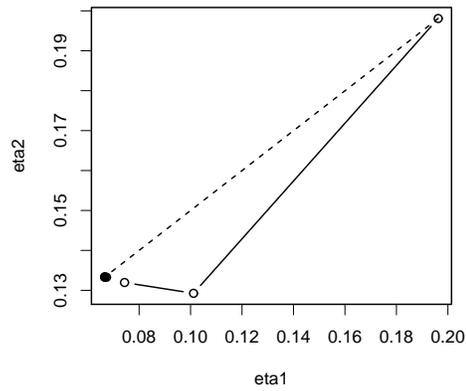}\\
  (b) $(\eta_1,\eta_2)$-plane.
  \end{tabular}
 \end{center}
 \caption{Trajectories of $\theta^{(k)}$ \SEI{for discrete algorithm (white circles)} and $\bar{\theta}(\tau)$
 for continuous algorithm \SEI{(dashed line)}: (a) $(\theta_1,\theta_2)$-plane
 and (b) $(\eta_1(\theta),\eta_2(\theta))$-plane,
 where $\eta_i(\theta)=\partial_i\log C(\theta)$ denotes
 the expectation parameter.
 \SEI{The black circle denotes the MLE.}}
 \label{fig:hgd}
\end{figure}

\section{Discussion}
\AK{
In this paper, we show that it is possible to perform statistical inference on Bingham distributions based on the exact maximum likelihood method. This is due to the fact that the normalising constants can be calculated accurately using the standard theory of holonomic functions. The only requirement for the algorithm to generate the correct value is to start from some exact initial point of the curve along which the final solution located. In our examples, the Taylor expansion method can be easily utilized to generate an acceptable starting point. For example, as shown in Section~\ref{subsection:HG}, one starting point could be $c(\theta/r)$ where  $r$ is such that the entries of the rescaled  vector $\theta/r$ are so small so that the Taylor expansion estimation can be very accurate at $\theta/r$. Alternatively, one could use the starting points given in Table~\ref{table:universal-value}. While the method proposed is rather more computationally demanding than the saddle point approximation, it is in fact very fast in the R package implementations and behaves well even for extreme values of the parameter $\theta$.  We also show how the algorithms can be easily adopted in the degenerate cases of multiplicities in the parameter vector $\theta$.
}

\appendix

\section*{Appendices}

\section{One-dimensional representation}

 First we briefly describe derivation of the one-dimensional representation (\ref{eq:Bingham-1dim})
 according to Kume and Wood (2005).
 Note that the parameter $\lambda_i$ in their paper
 is our $-\theta_i$.
 Consider $p$ independent normal random variables $x_i\sim N(0,(-2\theta_i)^{-1})$,
 where $\theta_i<0$ for all $i$.
 Then the marginal density of $r=\sum_{i=1}^px_i^2$ is directly calculated as
 \SEI{
 \begin{align}
  f(r) = \frac{r^{p/2-1}\prod_k\sqrt{-\theta_k}}{\Gamma(p/2)}\frac{C(r\theta)}{C(0)}.
  \label{eq:appendix-1}
 \end{align}
 }
 On the other hand, the characteristic function of $r$ is
 \begin{align}
  \phi(s)
  = E[e^{is\sum_ix_i^2}]
  = \prod_k\sqrt{\frac{-\theta_k}{-\theta_k-is}}.
  \label{eq:appendix-2}
 \end{align} 
 In general, the density function is represented by its characteristic function as
 \begin{align}
  f(r) = \frac{1}{2\pi}\lim_{\epsilon \to 0}\int_{-\infty}^{\infty}
  \phi(s)e^{-is r-\epsilon s^2/2}ds
  \label{eq:appendix-3}
 \end{align}
 at arbitrary continuous point of $f$
 (see e.g.\ \cite{Feller2}).
 By combining the equations (\ref{eq:appendix-1}) to (\ref{eq:appendix-3}), we have
 \SEI{
 \begin{align*}
  C(\theta)
  &= \frac{\Gamma(p/2)C(0)}{\prod_k\sqrt{-\theta_k}}f(1)
  \\
  &= \frac{C(0)}{2\pi}\lim_{\epsilon\to 0}\int_{-\infty}^{\infty}
  \frac{1}{\prod_k\sqrt{-\theta_k-is}}e^{-is-\epsilon s^2/2}ds.
 \end{align*}
 }Note that this expression holds for any $p\geq 2$.
 If $p\geq 3$, then
 \SEI{
 \begin{align*}
  C(\theta)
  &= \frac{C(0)}{2\pi}\int_{-\infty}^{\infty}
  \frac{1}{\prod_k\sqrt{-\theta_k-is}}e^{-is}ds
 \end{align*}
 }since $|\prod_{k=1}^p(-\theta_k-is)^{-1/2}|$ is integrable over $(-\infty,\infty)$.
 By analytic continuation with respect to $s$,
 we obtain
 \SEI{
 \begin{align}
  C(\theta)
  &= \frac{C(0)}{2\pi}\int_{-\infty}^{\infty}
  \frac{1}{\prod_k\sqrt{-\theta_k-t_0-is}}e^{-t_0-is}ds
  \label{eq:appendix-exact2}
 \end{align}
 }for any real number $t_0$ less than $\min_k(-\theta_k)$.
 Even if some $\theta_i$'s are not negative,
 the equation (\ref{eq:appendix-exact2}) still holds 
 due to analytic continuation with respect to $\theta$,
 as long as $t_0<\min_k(-\theta_k)$.
 Hence we obtain (\ref{eq:Bingham-1dim}).

\section{Truncation error of the power series} \label{section:truncation}

We derive the power series expansion of $C(\theta)$
and evaluate the truncation error according to \cite{yama4}.
\SEI{Let $\theta(\phi,d)=(\phi_1,\ldots,\phi_1,\ldots,\phi_q,\ldots,\phi_q)$ be a parameter vector
with multiplicities $(d_1,\ldots,d_q)$.}
\SEI{By \cite{KumeWood2007}},
we have
\SEI{
\begin{align*}
 &C(\theta(\phi,d))
 \\
 &= \int_{S^{p-1}} e^{\sum_{i=1}^p \theta_i(\phi,d) x_i^2} dx
 \\
 &= C(0)\sum_{k_1=0}^{\infty}\cdots\sum_{k_q=0}^{\infty}
 \frac{\phi_1^{k_1}\cdots\phi_q^{k_q}}{k_1!\cdots k_q!}
 \frac{\prod_i \Gamma(k_i+d_i/2)}{\Gamma(\sum_i (k_i+d_i/2))}
 \frac{\Gamma(\sum_i d_i/2)}{\prod_i\Gamma(d_i/2)}.
\end{align*}
}

Let
\SEI{
\begin{align*}
 &C_N(\theta(\phi,d))
 \\
 &:= C(0)
 \sum_{k_1+\cdots+k_q < N}
 \frac{\phi_1^{k_1}\cdots\phi_q^{k_q}}{k_1!\cdots k_q!}
 \frac{\prod_i \Gamma(k_i+d_i/2)}{\Gamma(\sum_i (k_i+d_i/2))}
 \frac{\Gamma(\sum_i d_i/2)}{\prod_i\Gamma(d_i/2)}.
\end{align*}
}

Then the truncation error is evaluated as
\SEI{
\begin{align*}
 &\frac{|C(\theta(\phi,d)) - C_N(\theta(\theta,d))|}{C(0)}
\\
 &\leq \sum_{k_1+\cdots+k_q\geq N}
 \frac{|\phi_1|^{k_1}\cdots|\phi_q|^{k_q}}{k_1!\cdots k_q!}
 \frac{\prod_i\Gamma(k_i+d_i/2)}{\Gamma(\sum_i(k_i+d_i/2))}
 \frac{\Gamma(\sum_i d_i/2)}{\prod_i\Gamma(d_i/2)}
 \\
 &\leq \sum_{k_1+\cdots+k_q\geq N}
 \frac{|\phi_1|^{k_1}\cdots|\phi_q|^{k_q}}{k_1!\cdots k_q!}
 \\
 &= \sum_{n=N}^{\infty}
 \frac{1}{n!}\sum_{k_1+\cdots+k_q=n}
 \frac{n!}{k_1!\cdots k_q!}
 |\phi_1|^{k_1}\cdots|\phi_q|^{k_q}
 \\
 &= \sum_{n=N}^{\infty}\frac{\|\phi\|_1^n}{n!}
 \\
 &\leq \frac{\|\phi\|_1^N}{N!}
 \sum_{n=N}^{\infty}\frac{\|\phi\|_1^{n-N}}{(N+1)^{n-N}}
 \\
 &= \frac{\|\phi\|_1^N}{N!}
 \frac{N+1}{N+1-\|\phi\|_1}.
\end{align*}
}

\section*{Acknowledgments}

The first author is supported by JSPS
Institutional Program for Young Researcher Overseas Visits.


\end{document}